\documentclass[12pt,journal,draftcls,letterpaper,onecolumn]{IEEEtran}
\usepackage{graphicx,times,cite,amsmath, amssymb, epsfig}
\usepackage{multirow}
\usepackage{subfig}
\newtheorem{theorem}{Theorem}
\usepackage[noend]{algorithmic}
\usepackage{algorithm}

\newlength{\figwidth}
\usepackage{color}

\begin{document}

\baselineskip 23.7pt
\title{Relay Selection and Performance Analysis in Multiple-User Networks}
\author{Saman Atapattu, Yindi Jing, Hai Jiang, and Chintha Tellambura
 \thanks{The authors are with the Department of Electrical and
   Computer Engineering, University of Alberta, Edmonton, AB T6G 2V4,
   Canada (e-mail: \{atapattu, yindi, hai1\}@ualberta.ca, chintha@ece.ualberta.ca).}}

\maketitle
\begin{abstract}
This paper investigates the relay selection (RS) problem in networks with multiple users and multiple common amplify-and-forward (AF) relays. Considering the overall quality-of-service of the network, we first specify our definition of optimal RS for multiple-user relay networks. Then an optimal RS (ORS) algorithm is provided, which is a straightforward extension of an RS scheme in the literature that maximizes the minimum end-to-end receive signal-to-noise ratio (SNR) of all users. The complexity of the ORS is quadratic in both the number of users and the number of relays. Then a suboptimal RS  (SRS) scheme is proposed, which has linear complexity in the number of relays and quadratic complexity in the number of users. Furthermore, diversity orders of both the ORS and the proposed SRS are theoretically derived and compared with those of a naive RS scheme and the single-user RS network. It is shown that the ORS achieves full diversity; while the diversity order of the SRS decreases with the the number of users. For two-user networks, the outage probabilities and array gains corresponding to the minimum SNR of the RS schemes are derived in closed forms. It is proved that the advantage of the SRS over the naive RS scheme increases as the number of relays in the network increases. Simulation results are provided to corroborate the analytical results.

\end{abstract}
\begin{keywords}
Array gain, diversity order, multiple-user networks, outage probability, relay selection.
\end{keywords}
\newpage
\section{Introduction}
Multiple-input-multiple-output (MIMO) techniques are effective ways of achieving spatial diversity in wireless communications. Since installing multiple antennas on wireless nodes is not always possible in practice, cooperative communication, a concept that takes advantage of the possible cooperation among multiple nodes in a network to form virtual MIMO configuration, has received significant attention in the wireless community \cite{Sendonaris2003a, Laneman2004}. 
For cooperative networks with multiple relays, relay selection (RS) is one important and effective technique because properly designed RS strategies can achieve full spatial diversity with low complexity and overhead.

RS problems have been extensively studied in the open literature for networks with single source-destination pair, referred to as {\it single-user networks, e.g.,} \cite{Bletsas06,Adve2007,Molisch2008,Karagiannidis2008,Yindi09a}.
Recently, there is increasing interest in relay networks with multiple source-destination pairs, referred to as {\it multiple-user networks}. Typical multiple-user networks include ad-hoc, sensor, and mesh networks. For a multiple-user multiple-relay network, proper RS is vital, however, limited attention has been paid to the RS problem. RS schemes proposed for single-user networks cannot be extended to multiple-user networks straightforwardly due to the challenges in the performance evaluation,  the competition among users, and the increased complexity  \cite{Yu07}. 

There are some research efforts on RS in multiple-user networks. A multiple-user multiple-relay network is considered with amplify-and-forward (AF) relaying and decode-and-forward (DF) relaying in \cite{Niu2010} and \cite{Ding2011}, respectively. Among the multiple users, the best user is first selected based on the quality of its direct link to its destination, then the selected user chooses the relay through which it can obtain the maximum end-to-end receive signal-to-noise ratio (SNR). Other users are not allowed to transmit. So in \cite{Niu2010} and \cite{Ding2011}, only one user with its best relay is selected at a time. \cite{Wittneben2010} considers a multi-user network, in which all relays are clustered into two groups based on the available channel state information (CSI). Only one relay group is selected, and all users communicate with their destinations through all relays in the selected group. In other words, multiple users and multiple relays are selected at a time, and each user is helped by all selected relays.

There is also some limited research work \cite{Nosratinia07,Xuemin2008,Sharma2010} on multiple-user multiple-relay networks in which at a time multiple users can transmit and each user is helped by a distinct relay set. Since the RS for one user may impact the choices of other users, the RS problem becomes more challenging.
In \cite{Nosratinia07}, grouping and partner selection for cooperative networks with DF relaying are considered. It investigates how to allocate relays to assist users and analyzes the effect of allocation policies on network performance. For each user, the relays are selected based on the strength of the user's channels to the relays. 
In \cite{Xuemin2008}, a single-user network is first considered.  Ensuring that relaying can achieve a larger channel capacity than direct transmission, a sufficient condition based on channel quality is derived to find a feasible set of relays for the single-user network. Then the work is extended to the multiple-user case, in which a semi-distributed RS is proposed to maximize the minimum capacity experienced by the users. However, the proposed scheme does not guarantee optimality because each user chooses a relay in its feasible set randomly. In \cite{Sharma2010}, an RS scheme that maximizes the minimum achievable data rate among all users is proposed. The complexity of the  scheme is linear in the number of users and quadratic in the number of relays. The work in \cite{Sharma2010} focuses on the proof of the optimality of the RS scheme. Analytical performance evaluation is not
provided. 

In this research, we consider a multiple-user multiple-relay network in which each user can only be helped by a single relay and one relay can help at most one user. The new contributions of this paper are listed as follows. \begin{itemize}
\item We specify an optimality measure of RS for multiple-user relay networks. Comparing with the previous used optimality, maximizing the minimum receive SNR among users, this measure guarantees the uniqueness of the optimal solution and takes into account the performance of all users in addition to the worst one. An optimal relay selection (ORS) scheme is provided, which is a straightforward extension of the minimum-SNR-maximizing RS scheme proposed in  \cite{Sharma2010}. The complexity of the ORS is quadratic in  both the number of users and the number of relays.
\item We propose a sub-optimal RS (SRS) scheme, whose complexity is linear in the number of relays and quadratic in the number of users. Thus, for networks with a large number of relays, the SRS is much faster than the ORS.
A naive RS scheme is also introduced as a benchmark, in which users select relays one by one based on their user indices.

\item For the ORS, the proposed SRS, and the naive RS, diversity orders are analyzed theoretically based on the minimum SNR among users using order statistics. It is shown that for a network with $N$ users and $N_r$ relays, with ORS, all users achieve diversity order $N_r$, which is the full diversity order of a single-user network with $N_r$ relays. Thus, user competition does not affect diversity order if optimally designed. For the SRS and naive RS, however, an achievable diversity order of all users is shown to be $N_r-N+1$.
\item For two-user networks, tight upper bounds on the outage probabilities of the ORS, SRS, and naive RS are derived. It is shown that the SRS achieves better array gain than the naive RS, and the advantage increases as there are more relays available in the network.
\item Numerically simulated outage probabilities are illustrated to justify our analytical results and show the advantage of the proposed SRS over the naive RS.

\end{itemize}

The rest of the paper is organized as follows. The system model and order statistics of receive SNRs are provided in Section \ref{sec-model}. RS schemes are introduced and discussed in Section \ref{sec-algo}, in which we introduce an ORS scheme, which is an extended version of the RS scheme in \cite{Sharma2010}, propose an SRS with lower complexity, and introduce a naive RS scheme for performance benchmark. The diversity orders of the ORS, the proposed SRS, and the naive RS schemes are analyzed in Section \ref{sec-div}. Outage probabilities of the three schemes in two-user networks are derived in closed-forms in Section \ref{sec-out}. Numerical results and the concluding remarks are presented in Section \ref{sec-simu} and Section \ref{sec-con}, respectively. 
\section{System Model and Order Statistics of Receive SNRs}
\label{sec-model}
Consider a wireless relay network with $N$ users sending information to their destinations via $N_r$ relay nodes, as shown in Fig. \ref{fig_system}.
Each node has a single antenna. The power budget is $P$ for each user and $Q$ for each relay. The fading coefficients from the $i$th user to the $j$th relay and from the $j$th relay to the $i$th destination are denoted as $f_{ij}$ and $g_{ji}$, respectively. There is no direct link between a user and its destination. All channels are assumed to be independent and identically distributed (i.i.d.) complex Gaussian fading with zero-mean and unit-variance, i.e., $f_{ij},g_{ji}\sim \mathcal{CN}(0,1)$. The channel amplitude thus follows Rayleigh distribution.  The RS schemes introduced and proposed in this paper are centralized. Thus a master node, which can be a destination, assumed to have perfect and global CSI, is in charge of the RS process.
To conduct AF relaying, the relays are assumed to know their channels with the users.

The users need the relays' help to send information to their destinations. In this paper, we assume that each user will be helped by one and only one relay to minimize the synchronization requirement on the network. Since multiple relays' participation consumes more power, it also has the potential of power-saving. We also assume that each relay can help at most one user. This is to avoid having too much load on one relay, which may prolong network lifetime \cite{Laneman2004,Xuemin2008}. Thus, we need $N_r\ge N$.

A conventional half-duplex two-step transmission protocol is used \cite{Laneman2004}. The first step is the transmissions from users to relays, and the second step is the transmissions from relays to destinations.
To avoid interference, the users are assigned orthogonal channels using frequency-division or time-division multiple access. Without loss of generality, the transmission of User $i$ helped by Relay $j$ is elaborated here. Denote the information symbol of User $i$ as $x_i$, which has unit average energy. Applying AF relaying with coherent power coefficient \cite{Atapattu_ChinaCom}, the receive signal at Destination $i$ can be written as
{\small \begin{equation}
\label{eq:RxDest}
\begin{split}
y_{ij}=  \sqrt{\frac{PQ}{P|f_{ij}|^2+1}}f_{ij}g_{ji}x_i+\sqrt{\frac{Q}{P|f_{ij}|^2+1}}g_{ji}n_{r_j}+n_{d_i}
\end{split}
\end{equation}}
where $n_{r_j}$ and $n_{d_i}$ are the additive noises at Relay $j$ and Destination $i$, respectively, which are assumed to be i.i.d.~following $\mathcal{CN}(0,1)$. The end-to-end receive SNR of User $i$ thus equals
{\small \begin{equation}
\label{eq:EndSnra}
\begin{split}
\gamma_{ij}=  \frac{PQ|f_{ij}g_{ji}|^2}{P|f_{ij}|^2+Q|g_{ji}|^2+1}.
\end{split}
\end{equation}}
For the simplicity of the presentation, in the rest of the paper, the SNR of a user means its end-to-end receive SNR.

\newpage

One main problem of this paper is to find RS schemes that lead to good performance. At the same time, low complexity is also desired for practical consideration. For the RS problem in single-user networks, the performance criterion is straightforward. Also, the competition is only among relays not users. In contrast, the RS problem in multiple-user networks is a lot more challenging: (i) the multiple communication tasks complicate the performance criterion specification and theoretical analysis; (ii) in addition to the competition among relays, there is competition among users to select their best relays in order to maximize their individual advantages; (iii) the complexity of exhaustive search is $\mathcal{O}(N_r^{N})$, which is very high for large networks. A good RS scheme should take into account the overall network quality-of-service, the fairness  among users, and the complexity.

In \cite{Sharma2010}, an RS scheme is proposed which maximizes the minimum transmission rate of the users, which is equivalent to maximizing the minimum SNR of  the users. With this RS criterion, however, the RS solution may not be unique and only the worst user's performance is optimized. This paper uses a modified design criterion, which is specified as follows. An RS solution is call {\it optimal} if it has the following properties:
\begin{itemize}
\item Property 1: the minimal SNR, denoted as $\gamma_{\min}$, among the users is maximized, which equivalently means that the minimum achievable data rate of all users is maximized and the maximum outage or error rate of all users is minimized;
\item Property $k$ ($k=2,...,N$): conditioned on the preceding $k-1$ properties, the $k$th minimal SNR of all user SNRs is maximized.
\end{itemize}
In contrast to maximizing the minimum receive SNR only, the new optimality definition guarantees the uniqueness of the solution and considers all users in addition to the worst one.

As to the performance measure, we consider outage probability, diversity order, and array gain. An outage occurs if the SNR drops below a predetermined SNR threshold $\gamma_{th}$. The outage probability corresponding to $\gamma_{\text{min}}$, denoted as $P_\text{out,upp}$,  is thus an upper bound on the outage probability of all users because their SNRs are always not lower than $\gamma_{\text{min}}$. Diversity order shows how fast the outage probability decreases with the increase in the transmit power in the high transmit power range. It is conventionally defined as $d\triangleq-\displaystyle\lim_{P\rightarrow\infty}\frac{\log P_\text{out}}{\log {P}}$ \cite{Wang2003} where $P_\text{out}$ is the outage probability. For the same reason, the diversity order derived based on $\gamma_{\text{min}}$ is a lower bound on the diversity orders of all users. When diversity orders of two designs are the same, a better measure for performance comparison is the array gain which is the difference between the required power levels of the two designs to reach the same outage level.

To help the RS procedure, we consider all relay choices for the users and construct a receive SNR matrix as
\begin{equation}
\label{eq:EndSnr}
{\bf {\Gamma}}= \left[\begin{array}{cccc} {\boldsymbol \gamma}_{1} & {\boldsymbol\gamma}_2 & \cdots & {\boldsymbol\gamma}_{N} \end{array}\right]^T
\end{equation}
where ${\boldsymbol\gamma}_{i}=\left[\begin{array}{cccc} \gamma_{i1} & \gamma_{i2} & \cdots, & \gamma_{iN_r} \end{array}\right]$ with $\gamma_{ij}$ the receive SNR of User $i$ helped by Relay $j$. ${\boldsymbol\gamma}_{i}$ contains the possible receive SNRs of User $i$ if helped by any of the $N_r$ relays. It is the $i$th row of ${\bf {\Gamma}}$. The $j$th column of ${\bf {\Gamma}}$ contains the receive SNRs of the users if Relay $j$ is chosen to help them. ${\bf {\Gamma}}$ is a $N\times N_r$ matrix.

Now we consider the statistics of the receive SNRs which will be used for theoretical analysis later. Since all channels are i.i.d., $\gamma_{ij}$'s are also i.i.d.. Denote their cumulative distribution function (CDF) and probability density function (PDF) as $F_{\gamma}(x)$ and $f_{\gamma}(x)$, respectively. From the results in \cite{Koyuncu2008}, we have
\begin{equation}
\label{eq:adfaf}
\begin{split}
F_{\gamma}(x)=1-2\sqrt{\frac{x(x+1)}{PQ}}e^{-\left(\frac{1}{P}+\frac{1}{Q}\right)x}
\mathcal{K}_1\left(2\sqrt{\frac{x(x+1)}{PQ}}\right),
\end{split}
\end{equation}
where $\mathcal{K}_1(\cdot)$ is the modified first-order Bessel function of the second kind. Since $x\mathcal{K}_1(x)\approx 1$ for small $x$ \cite{gradshteyn2007}, $F_\gamma(x)$ can be well-approximated for large $P$ and $Q$ as
\begin{equation}
\label{eq:cdfappx}
\begin{split}
F_\gamma(x)\approx 1-e^{-(\frac{1}{P}+\frac{1}{Q})x}=\left(\frac{1}{P}+\frac{1}{Q}\right)x-\sum_{i=2}^{\infty}\left[-\left(\frac{1}{P}+\frac{1}{Q}\right)\right]^i x^i.
\end{split}
\end{equation}

If we sort $\gamma_{ij}$'s in descending order as
\begin{equation}
\gamma_1 \geq \cdots \gamma_k \geq \cdots \geq \gamma_{NN_r}
\label{SNR-order}
\end{equation}
where $\gamma_k$ is the $k$th largest element of ${\bf \Gamma}$, and using the results of eq. (7)-(14) in \cite{Papoulis} of order statistics, the PDF of $\gamma_k$ can be given as
\begin{equation}
\label{eq:orderstatpdf}
f_{\gamma_k}(x)=\frac{(NN_r)!F_\gamma(x)^{NN_r-k}[1-F_\gamma(x)]^{k-1}f_\gamma(x)}{(NN_r-k)!(k-1)!}.
\end{equation}
Using binomial expansion, and subsequently applying integration by parts, the CDF of $\gamma_k$ can be derived as $F_{\gamma_k}(x)=\int_0^x f_{\gamma_k}(t) dt$ to yield
\begin{equation}
\label{eq:orderstatcdf}
F_{\gamma_k}(x)=\sum_{i=0}^{k-1}\frac{(NN_r)!\binom{k-1}{i}(-1)^iF_{\gamma}(x)^{NN_r-k+i+1}}{(NN_r-k+i+1)(NN_r-k)!(k-1)!}.
\end{equation}

\section{RS schemes}
\label{sec-algo}
For the $N$-user $N_r$-relay network, in \cite{Sharma2010}, an RS scheme is developed using the ``linear marking" mechanism to maximize the minimal SNR among all users, with worst-case complexity $\mathcal{O}(NN_r^2)$. Here we consider the following extension to the RS scheme in \cite{Sharma2010} to obtain the RS that is optimal with the definition specified in Section \ref{sec-model}. We first apply the RS scheme in \cite{Sharma2010} to find a solution that maximizes the minimal SNR. Suppose that the minimal SNR is with User $i$ and Relay $j$. Then we delete User $i$ from the user list and delete Relay $j$ from the relay list, and apply the RS scheme in \cite{Sharma2010} again to the remaining user list and relay list. This procedure is repeated until all users find their relays. Then we get an RS result that has the properties defined in Section \ref{sec-model}. 
We refer to it as the {\it optimal relay selection (ORS)} in the sequel.
The worst-case complexity of the ORS scheme is $\mathcal{O}(N^2N_r^2)$, which is quadratic in the number of users and in the number of relays.

For networks with a large number of relays, quadratic complexity in the number of relays may be undesirable. Thus, we also propose a {\it suboptimal relay selection} (SRS) scheme, described as Algorithm \ref{algo-1}, whose complexity is linear in the number of relays.

\begin{algorithm}[ht]
\caption{The Suboptimal RS (SRS) Scheme}
\label{algo-1}
\begin{algorithmic}[1]
\STATE Assign ${\bf \Gamma}_0={\bf \Gamma}$.
\FOR{$k=N:1$}
\STATE Let $k$ denote the number of rows in ${\bf \Gamma}_0$.
\STATE Find the maximum element of each row of ${\bf \Gamma}_0$. Denote the $k$ elements as $\gamma_{1j_1^*}, \cdots, \gamma_{kj_k^*}$.
\STATE Find $\gamma_{i^*j^*}=\text{min}\left(\gamma_{1j_1^*}, \cdots, \gamma_{kj_k^*}\right)$, and assign Relay $j^*$ to User $i^*$.
\STATE Delete the $j^*$th column and the $i^*$th row of ${\bf \Gamma}_0$.
\ENDFOR
\end{algorithmic}
\end{algorithm}

The main idea of the SRS is to find a relay for each user sequentially (not necessaries in the order of the user index) to achieve a complexity that is linear in the number of relays. In Step 4, the best relay for each user that has not selected a relay yet is found. To avoid RS conflict, in Step 5, the user with the smallest best SNR selects its best relay. This procedure is repeated until all users have made their selections.

Now we consider the worst-case complexity of the SRS scheme. If we consider the $k$th round of RS,  the required number of operations in Step 4 to find the maximum elements of $(N-(k-1))$ rows is $(N_r-k)(N-k+1)$; the required number of operations in Step 5 is $N-k$. Therefore, the total complexity for the SRS is
\begin{equation}
\label{e:complexity}
\begin{split}
\mathcal C  = \sum_{k=1}^{N} \left[(N_r-k)(N-k+1)+(N-k)\right]  = \frac{N(3NN_r+3N_r-N^2-5)}{6}.
\end{split}
\end{equation}
Noting that $N\le N_r$, from (\ref{e:complexity}), the complexity behaves as $\mathcal O\left(N^2N_r\right)$, linear in the number of relays and quadratic in the number of users. Therefore, for networks with many more relays than users, the SRS is advantageous in complexity.

The SRS does not always result in the optimal solution. When the best relays of two or more of the users are the same, the SRS scheme may lead to a suboptimal result. To see this, consider the following example of a network with two uses and four relays. For one channel realization, we have the SNR matrix: ${\bf \Gamma}= \left[ \begin{array}{ccccc}
1.08  &  0.14  &  0.09  &  0.05 \\
1.07 &  0.15  &  0.50  &  0.04 \end{array} \right]$. The ORS scheme selects Relay 1 for User 1 and Relay 3 for User 2, with SNR being $1.08$ and $0.5$ for the two users, respectively. This is the optimal RS solution. The SRS however selects Relay 2 for User 1 and Relay 1 for User 2, with the SNRs being $0.14$ and $1.07$ for the two users, respectively, which is not optimal.

In this section, we also introduce a naive RS scheme as a benchmark in evaluating the ORS and the SRS schemes. Intuitively, for the multi-user network, a naive method is to assign the best relays to the users
one by one from User 1 to  User $N$.
That is, User 1 first selects its best relay (the best relay results in the maximum SNR). Then User 2 selects its best relay among the remaining $N_r-1$ relays; and so on so forth until User $N$ selects its best relay among the remaining $N_r-N+1$ relays.  As to the complexity, $N_r-k$ operations are needed to find the best relays for User $k$. Thus, the overall complexity is $\sum_{k=1}^{N}(N_r-k)=\frac{1}{2}(2NN_r-N^2-N)$, which is linear in both the number of relays and the number of users. Obviously, the naive RS does not always result in the optimal RS.

\section{Diversity Order Analysis}
\label{sec-div}
In this section, we analyze the diversity orders of the schemes introduced in Section \ref{sec-algo}. As explained in Section \ref{sec-model},  we can derive achievable diversity orders of RS schemes based on the analysis of the outage probability corresponding to $\gamma_{\text{min}}$.

\subsection{Diversity Order of ORS}\label{s:diversity_order_ORS}
It is noteworthy that to our best knowledge, the performance analysis of the ORS is not available in the literature. This paper is the first that derives the diversity order of the ORS. The following theorem is proved.

\begin{theorem}\label{t:th5}
With ORS, each user achieves diversity order $N_r$. 
\end{theorem}
\begin{proof}
The receive SNR matrix for a general network is given in (\ref{eq:EndSnr}). Using the ORS and with the SNR ordering in (\ref{SNR-order}), $\gamma_{\text{min}}$ can take $\gamma_{N}, \ldots,$ or $\gamma_{(N-1)N_r+1}$. 
Thus an outage probability upper bound, $P_{\text{out,upp,ORS}}$, which is the outage probability with respect to $\gamma_{\min}$, can be calculated as
$P_{\text{out,upp,ORS}}  = \sum_{k=N}^{(N-1)N_r+1} \text{Prob}(\gamma_{\text{min}}= \gamma_k)F_{\gamma_k}(\gamma_{th})$.

For simplicity of the presentation, we assume $Q=P$ in the sequel.
The results can be generalized to unequal power case straightforwardly as long as the powers of all nodes have the same scaling.
When $P$ is large, using  (\ref{eq:cdfappx}) and (\ref{eq:orderstatcdf}), we have the following approximation
{\small \begin{equation*}
P_{\text{out,upp,ORS}}\approx\hspace{-3mm} \sum_{k=N}^{(N-1)N_r+1}\sum_{i=0}^{k-1} \frac{\text{Prob}(\gamma_{\text{min}}= \gamma_k)(NN_r)!\binom{k-1}{i}(-1)^i\left(\frac{2\gamma_{th}}{P}-\sum_{j=2}^{\infty}\frac{(-1)^j(2\gamma_{th})^j}{P^j}\right)^{NN_r-k+i+1}}{(NN_r-k+i+1)(NN_r-k)!(k-1)!}.
\end{equation*}}
Since $\gamma_{i,j}$'s are i.i.d., $\text{Prob}\left(\gamma_{\text{min}}= \gamma_{(N-1)N_r+1}\right)$ does not depend on $P$. Thus, with respect to $P$, the highest order term in the summation is the term with $k=(N-1)N_r+1$ and $i=0$. Thus, we have
\begin{equation}
\label{eq:out_gen_opt2}
P_{\text{out,upp,ORS}}\approx \frac{\text{Prob}\left(\gamma_{\text{min}}= \gamma_{(N-1)N_r+1}\right)\left(2\gamma_{th}\right)^{N_r}(NN_r)!}{N_r(N_r-1)!(NN_r-N_r)!}P^{-N_r}
+\mathcal{O}\left(P^{-(N_r+1)}\right)
\end{equation}
which has diversity $N_r$. Since the outage probability of each user is no higher than $P_{\text{out,upp,ORS}}$, we conclude that each user achieves diversity order $N_r$.
\end{proof}

For a single-user network with $N_r$ relays, the best RS achieves diversity $N_r$. Theorem \ref{t:th5} shows that for multiple-user networks, even with user competition for relays, ORS can achieve full single-user diversity order.

However, due to user competition, the achievable array gain of the ORS for multiple-user network is smaller compared with that of the single-user case. To see this, we investigate the probability of  $\gamma_{\text{min}}=\gamma_{(N-1)N_r+1}$. The event happens when the largest $(N-1)N_r$ elements,  i.e., $\gamma_1, \ldots, \gamma_{(N-1)N_r}$, of ${\bf \Gamma}$ given in (\ref{eq:EndSnr}), are in $N-1$ rows of the $N$ rows in ${\bf \Gamma}$. In other words,
the smallest $N_r$ elements of ${\bf \Gamma}$, i.e., $\gamma_{(N-1)N_r+1}, \ldots, \gamma_{NN_r}$, are all in the same row. This happens with the probability
\begin{equation}
\label{eq:pmin_prob_opt}
\begin{split}
\text{Prob}\left(\gamma_{\text{min}}= \gamma_{(N-1)N_r+1}\right)=\frac{(N_r-1)!}{\prod_{l=1}^{N_r-1}(NN_r-l)}.
\end{split}
\end{equation}
Using  (\ref{eq:out_gen_opt2}) and (\ref{eq:pmin_prob_opt}), $P_{\text{out,upp,ORS}}$ can be simplified as
\begin{equation}
\label{eq:out_gen_opt3}
\begin{split}
P_{\text{out,upp,ORS}} \approx N\left(2\gamma_{th}\right)^{N_r}P^{-N_r}+\mathcal{O}\left( P^{-(N_r+1)}\right),
\end{split}
\end{equation}
which increases as $N$ increases. This shows that due to the competition among the $N$ users, the achievable array gain of the multi-user network degrades linearly with the number of users compared with the single-user network.

\subsection{Diversity Order of SRS}\label{s:diversity_order_SRS}
For the SRS proposed in Section \ref{sec-algo}, the following diversity order result is proved.

\begin{theorem}\label{t:th6}
With the SRS, the achievable diversity order of each user is no less than  $N_r-N+1$. 
\end{theorem}
\begin{proof}
With the SRS described in Algorithm \ref{algo-1}, $\gamma_{\text{min}}$ can take {\small $\gamma_{N}, \gamma_{N+1},\ldots,$ or $\gamma_{(N-1)(N_r+1)+1}$.} 
Thus, the outage probability can be calculated as $P_{\text{out,upp,SRS}}  = \sum_{k=N}^{(N-1)(N_r+1)+1} \text{Prob}(\gamma_{\text{min}}= \gamma_k)F_{\gamma_k}(\gamma_{th})$.

Similar to (\ref{eq:out_gen_opt2}), when $Q=P\gg1$, $P_{\text{out,upp,SRS}}$ can be approximated as
\begin{equation}
\label{eq:out_gen_semi2}
\begin{split}
P_{\text{out,upp,SRS}}\hspace{-1mm} \approx\hspace{-1mm} \frac{\text{Prob}\left(\gamma_{\text{min}}= \gamma_{(N-1)(N_r+1)+1}\right)\left(2\gamma_{th}\right)^{N_r-N+1}(NN_r)!P^{-(N_r-N+1)}}{(N_r-N+1)(N_r-N)!\left((N-1)(N_r+1)\right)!}
+\mathcal{O}\hspace{-1mm}\left(P^{-(N_r-N+2)}\right).
\end{split}
\end{equation}
Similarly, since $\text{Prob}\left(\gamma_{\text{min}}= \gamma_{(N-1)(N_r+1)+1}\right)$ does not depend on $P$, an achievable diversity order of every user is $N_r-N+1$. 
\end{proof}

\subsection{Diversity Order of Naive RS}
\label{s:diversity_order_naive}

Now, we analyze the diversity order of the naive RS. We first consider the relay assignment to User $k$ ($k\in \{1,\cdots,N\}$). User $k$ selects the best relay that results in the maximum SNR from $N_r-k+1$ available relays. Denote the maximum SNR of User $k$ as $\gamma_{\text{max},k}$. Since all SNRs, $\gamma_{ij}$'s, are i.i.d., the CDF of $\gamma_{\text{max},k}$ is thus $F_{\gamma_{\text{max},k}}(t)=[F_{\gamma_{ij}}(t)]^{N_r-(k-1)}$. The minimum SNR of the users is $\gamma_{\text{min}}=\min_{k= 1,\ldots,N}\{\gamma_{\text{max},k}\}$. The CDF of $\gamma_{\text{min}}$  is thus $F_{\text{min}}(t)=1-\prod_{k=1}^{N}\left[1-F_{\gamma_{\text{max},k}}(t)\right]$.  Therefore, an upper bound on the outage probability for the naive RS scheme is
\begin{equation}
\label{eq:out_gen_naive1}
\begin{split}
P_{\text{out,upp,naive}} = 1-\prod_{k=1}^{N}\left[1-F_{\gamma}(\gamma_{\text{th}})^{N_r-(k-1)}\right].
\end{split}
\end{equation}
When $Q=P\gg 1$, we have the following approximation
\begin{equation}
\label{eq:out_gen_naive2}
\begin{split}
P_{\text{out,upp,naive}} \approx \left(2\gamma_{th}\right)^{N_r-N+1} P^{-(N_r-N+1)}+ \mathcal{O}\left(P^{-(N_r-N+2)}\right).
\end{split}
\end{equation}
Since the outage probability of each user is no higher than $P_{\text{out,upp,naive}}$, an achievable diversity order of all users is $N_r-N+1$.

\section{Outage Probability Analysis for Two-User Networks}
\label{sec-out}

In this section, we provide outage probability analysis for two-user relay networks with $N_r \ge 2$ relays. With two users, the receive SNR matrix of the network can be written as
\begin{equation}
\label{eq:EndSnr2nr}
{\bf \Gamma}= \left[ \begin{array}{ccccc}
\gamma_{11} & \cdots & \gamma_{1j}&...& \gamma_{1N_r} \\
\gamma_{21}& \cdots & \gamma_{2j}&...& \gamma_{2N_r} \end{array} \right]_{2\times N_r}.
\end{equation}
%
\subsection{Outage Probability Bound of ORS}\label{s:outage_bound_ORS}
As mentioned in Section II, we calculate the outage probability based on the minimum SNR, $\gamma_{\text{min}}$, which provides an upper bound on both users' outage probabilities. The following theorem is proved.
\begin{theorem}\label{t:th1} For a two-user network, with the ORS, the outage probabilities of both users are upper bounded by
\begin{equation}
\label{eq:outalgo3}
\begin{split}
 P_{\text{out,upp,ORS}} =  \frac{N_r-1}{2N_r-1} F_{\gamma_2}(\gamma_{th})+\frac{N_r+2}{2(2N_r-1)}F_{\gamma_3}(\gamma_{th})  +\sum_{i=4}^{N_r+1}\frac{2N_r\binom{N_r}{i-1}}{(2N_r-(i-1))\binom{2N_r}{i-1}}F_{\gamma_i}(\gamma_{th}),
\end{split}
\end{equation}
\end{theorem}
where $F_{\gamma_k}(x)$ is the CDF of $\gamma_k$ given in (\ref{eq:orderstatcdf}).
\begin{proof}
With the ORS, $\gamma_{\text{min}}$ can take $\gamma_2, \gamma_3, \cdots,$ or $\gamma_{N_r+1}$. 
The outage probability upper bound, $P_{\text{out,upp,ORS}}$, can be calculated as
\setlength{\arraycolsep}{0pt}
\begin{eqnarray}
&&P_{\text{out,upp,ORS}}  =\text{Prob}(\gamma_{\text{min}} \leq\gamma_{th})= \sum_{k=2}^{K} \text{Prob}(\gamma_{\text{min}} = \gamma_k)\text{Prob}(\gamma_k \leq\gamma_{th})\nonumber \\
& =& \sum_{k=2}^{K} \text{Prob}(\gamma_{\text{min}} = \gamma_k)F_{\gamma_k}(\gamma_{th})
\label{eq:outage2}
\end{eqnarray}
\setlength{\arraycolsep}{5pt}
where $K=N_r+1$. We now calculate the probability of  $\gamma_{\text{min}}=\gamma_k$ where $k=2, \cdots, N_r+1$ by considering the following three cases.

\begin{itemize}
\item
$\gamma_{\text{min}}=\gamma_2$ happens when $\gamma_1$ and $\gamma_2$ are in two distinct rows and columns. Thus $\text{Prob}(\gamma_{\text{min}} = \gamma_2)=\frac{N_r-1}{2N_r-1}$.
\item
$\gamma_{\text{min}}=\gamma_3$ happens when $\gamma_1$ and $\gamma_2$ are in the same column, or $\gamma_1$ and $\gamma_2$ are in the same row and $\gamma_3$ is in a different row. Thus $\text{Prob}(\gamma_{\text{min}} = \gamma_3)=\frac{1}{2N_r-1}+\frac{N_r}{2(2N_r-1)}=\frac{N_r+2}{2(2N_r-1)}$.
\item
$\gamma_{\text{min}}=\gamma_k$ for $k=4, \cdots, N_r+1$ happens when all $\gamma_1, \gamma_2, \cdots, \gamma_{k-1}$ are in the same row and $\gamma_k$ is in a different row. Then $\text{Prob}(\gamma_{\text{min}} = \gamma_k)=\frac{2N_r\binom{N_r}{k-1}}{(2N_r-(k-1))\binom{2N_r}{k-1}}$. 
\end{itemize}
Using these probabilities in (\ref{eq:outage2}), (\ref{eq:outalgo3}) can be obtained.
\end{proof}

Using (\ref{eq:orderstatcdf}) and with some straightforward algebraic manipulations, (\ref{eq:outalgo3}) can be rewritten as
\begin{equation}
\label{eq:outalgo3a}
\begin{split}
 P_{\text{out,upp,ORS}} = & F_\gamma(\gamma_{th})^{N_r}\biggl[\frac{(N_r-1)(2N_r)!}{(2N_r-1)(2N_r-2)!}  \sum_{i=0}^{1}\frac{\binom{1}{i}(-1)^iF_\gamma(\gamma_{th})^{N_r+i-1}}{2N_r+i-1} \\ &
+\frac{(N_r+2)(2N_r)!}{4(2N_r-1)(2N_r-3)!}\sum_{i=0}^{2}\frac{\binom{2}{i}(-1)^iF_\gamma(\gamma_{th})^{N_r+i-2} }{2N_r+i-2}
\\ &
+\sum_{j=4}^{N_r+1}\sum_{i=0}^{j-1}\frac{2N_r(2N_r)!\binom{N_r}{j-1}\binom{j-1}{i}(-1)^iF_\gamma(\gamma_{th})^{N_r-j+i+1}}{(2N_r-j+1)(2N_r-j+i+1)(2N_r-j)!(j-1)!\binom{2N_r}{j-1}}
\biggr].
\end{split}
\end{equation}

Now we consider the large-power approximation of the outage probability for the special case that $Q=P$. This is useful in the array gain discussion in Section \ref{s:arraygain}.

When $N_r>2$, the highest order term of $P$ in (\ref{eq:outalgo3a}) is the term with $j=N_r+1$ and $i=0$ in the double summation, which equals 2. Therefore, using (\ref{eq:cdfappx}), $P_{\text{out,upp,ORS}}$ can be approximated for large $P$ as
\begin{equation}
\label{eq:outalgo3b}
\begin{split}
P_{\text{out,upp,ORS}} \approx 2^{N_r+1}\gamma_{th}^{N_r}P^{-N_r}+ \mathcal{O}\left( P^{-(N_r+1)}\right).
\end{split}
\end{equation}
When $N_r=2$, the double summation in (\ref{eq:outalgo3a}) does not appear and the highest order term in (\ref{eq:outalgo3a}) is the term with $i=0$ in the second summation, which equals 4. Thus,
\begin{equation}\label{e:outage_bound_ORS_Nequalto2}
P_{\text{out,upp,ORS}} \approx 2^{N_r+2}\gamma_{th}^{N_r}P^{-N_r}+ \mathcal{O}\left(P^{-(N_r+1)}\right).
\end{equation}


\subsection{Outage Probability Bound of SRS}\label{s:algo2}

For the SRS, we calculate the outage probability based on $\gamma_{\text{min}}$ similarly and obtain the following theorem.
\begin{theorem}\label{t:th3}
 For a two-user network, with the SRS, the outage probabilities of both user in the network are upper bounded by
\begin{equation}
\label{eq:outalgo4}
\begin{split}
P_{\text{out,upp,SRS}}  = & \frac{N_r-1}{2N_r-1} F_{\gamma_2}(\gamma_{th})+\frac{N_r+1}{2(2N_r-1)}F_{\gamma_3}(\gamma_{th}) +\sum_{i=4}^{N_r+1}\frac{2N_r\binom{N_r}{k-1}}{(2N_r-(k-1))\binom{2N_r}{k-1}}F_{\gamma_i}(\gamma_{th})\\ &  +\sum_{i=4}^{N_r+2}\frac{2(N_r-1)\binom{N_r}{k-2}F_{\gamma_i}(\gamma_{th})}{(2N_r-(k-2))(2N_r-(k-1))\binom{2N_r}{k-2}}.
\end{split}
\end{equation}
\end{theorem}
\begin{proof}
With the SRS, $\gamma_{\text{min}}$ can take $\gamma_2, \gamma_3, \cdots,$ or $\gamma_{N_r+2}$. Therefore, outage probability can be written as (\ref{eq:outage2})  where $K=N_r+2$. 
In the following, we calculate the probability of $\gamma_{\text{min}}=\gamma_k$ where $k=2, \cdots, N_r+2$. 

\begin{itemize}

\item
$\gamma_{\text{min}} = \gamma_2$ happens when $\gamma_1$ and $\gamma_2$ are in two distinct rows and columns. Thus $\text{Prob}(\gamma_{\text{min}} = \gamma_2)=\frac{N_r-1}{2N_r-1}$.

\item
$\gamma_{\text{min}} = \gamma_3$ happens when $\gamma_1$ and $\gamma_2$ are in the same column and $\gamma_3$ is in $\gamma_1$'s row, or $\gamma_1$ and $\gamma_2$ are in the same row and $\gamma_3$ is in a different row. Thus $\text{Prob}(\gamma_{\text{min}} = \gamma_3)=\frac{1}{2(2N_r-1)}+\frac{N_r}{2(2N_r-1)}=\frac{N_r+1}{2(2N_r-1)}$.

\item
$\gamma_{\text{min}}=\gamma_k$  for $k=4, \cdots, N_r+2$ happens when 1) $\gamma_2, \cdots, \gamma_{k-1}$ are in the same row, and $\gamma_1$ and $\gamma_k$ are in the other row, and $\gamma_1$ and $\gamma_2$ are in the same column. This event happens with the probability $\frac{2(N_r-1)\binom{N_r}{k-2}}{(2N_r-(k-2))(2N_r-(k-1))\binom{2N_r}{k-2}}$; or 2) $\gamma_1, \cdots, \gamma_{k-1}$ are in the same row and $\gamma_k$ is in the other row. This event happens with the probability $\frac{2N_r\binom{N_r}{k-1}}{(2N_r-(k-1))\binom{2N_r}{k-1}}$.

\end{itemize}
Using these probabilities in (\ref{eq:outage2}), (\ref{eq:outalgo4}) is obtained.
\end{proof}

Following the same steps in Section \ref{s:outage_bound_ORS}, $P_{\text{out,upp,SRS}}$ can be rewritten as 
\begin{equation}
\label{eq:outalgo4a}
\begin{split}
 P_{\text{out,upp,SRS}} = & F_\gamma(\gamma_{th})^{N_r-1}\biggl[\frac{({N_r}+1) (2 {N_r})!}{4 (2 {N_r}-1)(2{N_r}-3)!}\sum_{i=0}^{1}\frac{\binom{1}{i}(-1)^iF_\gamma(\gamma_{th})^{N_r+i}}{2N_r+i-1}\\ &
 +  \frac{({N_r}+1) (2 {N_r})!}{4 (2 {N_r}-1) (2 {N_r}-3)!}\sum_{i=0}^{2}\frac{\binom{2}{i}(-1)^iF_\gamma(\gamma_{th})^{N_r+i-1}}{2N_r+i-2}\\ & + \sum_{i=4}^{N_r+1}\sum_{j=0}^{i-1}\frac{2N_r(2N_r)!\binom{N_r}{i-1}\binom{i-1}{j}(-1)^jF_\gamma(\gamma_{th})^{N_r-i+j+2}}{(2N_r-i+1)(2N_r-i+j+1)(2N_r-i)!(i-1)!\binom{2N_r}{i-1}}\\ &\hspace{-10mm}
+ \sum_{i=4}^{N_r+2}\sum_{j=0}^{i-1}\frac{2(N_r-1)(2N_r)!\binom{N_r}{i-2}\binom{i-1}{j}(-1)^jF_\gamma(\gamma_{th})^{N_r-i+j+2}}{(2N_r-i+1)(2N_r-i+2)(2N_r-i+j+1)(2N_r-i)!(i-1)!\binom{2N_r}{i-2}}
\biggr].
\end{split}
\end{equation}

Next, we consider the large-power approximation of the outage probability for the special case that $Q=P$. Noticing that $F_{\gamma}(\gamma_{th})\approx \mathcal{O}\left(1/P\right)$, the highest order term of $P$ in (\ref{eq:outalgo4a}) is the term with $i=N_r+1$ and $j=0$ in the second double summation, which equals $2/(N_r+1)$. Therefore, using (\ref{eq:cdfappx}), $P_{\text{out,upp,SRS}}$ can be approximated for large $P$ as
\begin{equation}
\label{eq:outalgo4b}
\begin{split}
P_{\text{out,upp,SRS}} \approx \frac{2^{N_r}\gamma_{th}^{N_r-1}}{N_r+1} P^{-(N_r-1)}+ \mathcal{O}\left( P^{-N_r}\right).
\end{split}
\end{equation}

\subsection{Outage Probability Bound of  Naive RS}


Now, we consider the naive RS scheme. Using (\ref{eq:out_gen_naive1}), for two-user relay networks,
the CDF of $\gamma_{\text{min}}$ of the naive RS scheme is
$F_{\gamma_{\text{min}}}=[F_{\gamma_{ij}}(t)]^{N_r}+[F_{\gamma_{ij}}(t)]^{N_r-1}-[F_{\gamma_{ij}}(t)]^{2N_r-1}$. An upper bound on the outage probability for the naive RS scheme is thus
\begin{equation}
\label{eq:outnaivea}
\begin{split}
P_{\text{out,upp,naive}}  = F_\gamma(\gamma_{th})^{N_r-1}\biggl[1+F_\gamma(\gamma_{th})-F_\gamma(\gamma_{th})^{N_r}\biggr].
\end{split}
\end{equation}
When $Q=P\gg 1$, we have the following approximation
\begin{equation}
\label{eq:outnaiveb}
\begin{split}
P_{\text{out,upp,naive}} \approx (2\gamma_{th})^{N_r-1} P^{-(N_r-1)}
+ \mathcal{O}\left( P^{-N_r}\right).
\end{split}
\end{equation}

\subsection{Discussions}
\label{s:arraygain}
In this subsection, for the two-user network, we discuss the properties of the ORS and SRS schemes and compare with the naive scheme (the benchmark).

First, the ORS scheme is shown to produce the optimal RS result and full diversity with a complexity that is quadratic in the number of relays. We can compare its performance with single-user $N_r$-relay network to see the performance degradation due to the competition between the two users.
 For single-user network, with the best RS \cite{Adve2006}, the outage probability is $P_{\text{out,single}}=F_\gamma(\gamma_{th})^{N_r}$. For large $P$, it can be approximated as \[P_{\text{out,single}}\approx (2\gamma_{th})^{N_r} P^{-N_r}+ \mathcal{O}\left( P^{-(N_r+1)}\right).\]
 Now we compare the array gains of the single-user network and the two-user network. Considering the ratio of the outage probability upper bounds, we have
\begin{equation}
\label{eq:arraygain1}
\begin{split}
c_\text{ORS, single}=\displaystyle\lim_{P\rightarrow\infty}\frac{P_\text{out,upp,ORS}}{P_\text{out,single}}
=\left\{ \begin{array}{ll} 2\approx 3\,\text{dB} & \mbox{if $N_r>2$}\\
4\approx 6\,\text{dB} & \mbox{if $N_r=2$} \end{array}\right..
\end{split}
\end{equation}
This shows the degradation of performance of a two-user network due to the competition between the two users.
Compared with the naive RS scheme, diversity order results in Section \ref{sec-div} show that the ORS achieves a larger diversity order with higher complexity.

Second, we discuss the properties of the SRS. The SRS is suboptimal and it loses one diversity order in two-user relay networks based on the results in Section \ref{sec-div}. But it has a lower complexity: linear in the number of relays. Comparing with the naive RS scheme, the SRS has the same diversity. Now, we discuss the array gain difference of the SRS and the naive RS scheme. Considering the ratio of the outage probabilities using (\ref{eq:outalgo4b}) and (\ref{eq:outnaiveb}), we have
\begin{equation}
\label{eq:arraygain2}
\begin{split}
c_\text{naive, SRS}=\displaystyle\lim_{P\rightarrow\infty}\frac{P_\text{out,upp,naive}}{P_\text{out,upp,SRS}}=10\log \left(\frac{N_r+1}{2}\right) \text{dB}.
\end{split}
\end{equation}
This shows that the SRS has a larger array gain due to a clever order of users in selecting relays. As there are more relays in the network, the array gain advantage of SRS increases in the logarithm of the relay number.


\section{Numerical and Simulation Results}
\label{sec-simu}

In this section, we give simulation results to justify our analysis, and to evaluate the performance of the ORS, SRS, and naive RS schemes.
All nodes are assumed to have the same power, i.e., $Q=P$. The SNR threshold $\gamma_{th}$ is set to be  5dB.

Fig.~\ref{fig_verify} is on two-user networks with two and four relays. It shows simulated outage probability corresponding to $\gamma_{\min}$ (shown in small circles), exact analytical outage probability corresponding to $\gamma_{\min}$ in eqs. (\ref{eq:outalgo3a}), (\ref{eq:outalgo4a}) and (\ref{eq:outnaivea}) (shown in continuous lines), and approximated analytical outage probability corresponding to $\gamma_{\min}$ in eqs. (\ref{eq:outalgo3b}), (\ref{e:outage_bound_ORS_Nequalto2}), (\ref{eq:outalgo4b}) and (\ref{eq:outnaiveb}) (shown in dashed lines) for the ORS, SRS, and naive RS schemes. 
For the entire simulated power range, we can see that our analytical results exactly match the simulation results for all schemes and both network settings. This confirms the accuracy of our analysis. The outage probability approximations  are accurate for large $P$. This confirms the validity of our analysis in diversity order and array gain.

In Fig.~\ref{fig_Alg1Alg2}, for a two-user network with two and four relays, we show the simulated outage probabilities of User 1 with the ORS and SRS schemes and compare with the outage upper bound derived using the minimum SNR. Due to the homogeneity of the network, User 2 has the same outage probability as User 1. It can be seen from the figure that the outage probability upper bounds are tight especially when the number of relays is large. It is within 2 dB and 1 dB of the user's outage probability for $N_r=2$ and $N_r=4$, respectively.
It can be further observed that ORS and SRS have almost the same performance at low transmit power region, but ORS has better performance in the high transmit power region because of its\ diversity advantage.

In Fig.~\ref{fig_allScheme1}, we compare the simulated outage probabilities of the ORS and SRS with those of the naive RS scheme and a random RS scheme in a two-user network with 2 or 4 relays. In the random RS, each user randomly chooses a relay to help without conflict.
For the ORS, the SRS, and the random RS scheme, both users have the same outage probability, thus only the outage probability of User 1 is shown. For the naive RS scheme, outage probabilities of User 1 and User 2 are different. User 1 actually achieves the performance of the single-user case since it has all $N_r$ relays to choose from. User 2 has a worse performance since it has only $N_r-1$ relays to choose from. Consider $N_r=4$. It can be seen from Fig.~\ref{fig_allScheme1} that the random RS has diversity order 1 only, while the achieved diversity orders of ORS and SRS  are 4 and 3, respectively. The naive RS scheme also achieves diversity order 3. Further, we compare the outage probability of User 1 in ORS denoted $P_\text{out,1,ORS}$ with that of User 1 in the naive RS scheme denoted $P_\text{out,1,naive}$ (which is equivalent to the outage probability of a single-user best-relay case denoted $P_\text{out,single}$). Since $P_\text{out,1,ORS}$ is smaller than $P_\text{out,upp,ORS}$, we can conclude that $\frac{P_\text{out,1,ORS}}{P_\text{out,1,naive}}$ is less than $c_\text{ORS, single}$ given in (\ref{eq:arraygain1}). From Fig.~\ref{fig_allScheme1}, in the high transit power region, $\frac{P_\text{out,1,ORS}}{P_\text{out,1,naive}}$ is around 3 dB when $N_r=2$ and is almost 0 dB when $N_r=4$. This means with more relays, the performance of either user becomes closer to the single-user best-relay case.
Next we compare the outage probability of User 1 in SRS denoted $P_\text{out,1,SRS}$ with that of User 2 (the worse user) in the naive RS scheme denoted $P_\text{out,2,naive}$. Note that $P_\text{out,upp,SRS}$ and $P_\text{out,upp,naive}$ are upper bounds of $P_\text{out,1,SRS}$ and $P_\text{out,2,naive}$, respectively. So $\frac{P_\text{out,2,naive}}{P_\text{out,1,SRS}}$ may not be equal to $c_\text{naive, SRS}$ given in (\ref{eq:arraygain2}) but an approximation. From Fig.~\ref{fig_allScheme1}, in high transmit power region, $\frac{P_\text{out,2,naive}}{P_\text{out,1,SRS}}$ is around 4.7 dB when $N_r=2$ and around 6 dB when $N_r=4$.


Next we further investigate the array gain differences 1) between ORS and the single-user best-relay case (which is equivalent to the performance of User 1 in the naive RS scheme), and 2) between the SRS scheme and the naive RS scheme. Fig.~\ref{fig_ArrNr} (in log-log scale) shows simulation results of the outage probability bounds of ORS, SRS, and naive RS ($P_\text{out,upp,ORS}$, $P_\text{out,upp,ORS}$, $P_\text{out,upp,naive}$) and the outage probability of User 1 in the naive RS scheme (equivalent to $P_\text{out,single}$). We have the following observations from the figure. Compared with the single-user best-relay case in high transmit power region, the ORS has 6 dB loss in array gain when $N_r=2$, and has 3 dB loss in array gain when $N_r=4$. Compared with the SRS in terms of outage probability bound in high transmit power region, the naive RS scheme has 1.7 dB loss in array gain when $N_r=2$, and has 4 dB loss in array gain when $N_r=4$. These are consistent with our analysis in (\ref{eq:arraygain1}) and (\ref{eq:arraygain2}).


Figs. \ref{fig_multiuserNr4} and \ref{fig_multiuserNr6} are on three-user ($N=3$) networks with four and six relays, respectively. The two figures show the simulated outage probabilities of the users with ORS, SRS, naive, and random RS schemes.  We can see that ORS achieves diversity order 4 and 6 in the two cases, respectively, which are full diversities; SRS achieves diversity order 2 and 4 in the two cases, respectively, which are equal to ($N_r-N+1$). For the naive RS scheme, outage probabilities of User 1, User 2 and User 3 are different, and they have diversity orders 4, 3, and 2, respectively, for $N_r=4$, and 6, 5 and 4, respectively, for $N_r=6$.  Again, it can be seen that the random selection has diversity order 1 only. These observations on diversity orders confirm the validity of the approximations in the diversity order analysis in Section \ref{sec-div}.
It also shows that the proposed ORS and SRS have better fairness among users. 
\section{Conclusion}
\label{sec-con}
The relay selection problem in a network with multiple users and multiple AF relays is investigated. First, a scheme that can achieve the optimal relay selection result with complexity quadratic in number of users and in number of relays is introduced. Then a suboptimal relay selection scheme is proposed, with complexity quadratic in number of users and linear in number of relays. 
The diversity orders of the schemes are theoretically derived.
For two-user networks, outage probabilities corresponding to the minimal SNR of different relay selection schemes are theoretically derived. The suboptimal relay selection is shown to achieve a higher array gain than a naive relay selection.



\begin{thebibliography}{10}
\bibitem{Sendonaris2003a}
A.~Sendonaris, E.~Erkip, and B.~Aazhang, ``User cooperation diversity. {P}art
  {I}: System description,'' \emph{{IEEE} Trans. Commun.}, vol.~51, no.~11, pp.
  1927--1938, Nov. 2003.

\bibitem{Laneman2004}
J.~N. Laneman, D.~N.~C. Tse, and G.~W. Wornell, ``Cooperative diversity in
  wireless networks: Efficient protocols and outage behavior,'' \emph{{IEEE}
  Trans. Inform. Theory}, vol.~50, no.~12, pp. 3062--3080, Dec. 2004.

\bibitem{Bletsas06}
A.~Bletsas, A.~Khisti, D.~Reed, and A.~Lippman, ``A simple cooperative
  diversity method based on network path selection,'' \emph{{IEEE} J. Select.
  Areas Commun.}, vol.~24, no.~3, pp. 659--672, Mar. 2006.

\bibitem{Adve2007}
Y.~Zhao, R.~Adve, and T.~Lim, ``Improving amplify-and-forward relay networks:
  optimal power allocation versus selection,'' \emph{{IEEE} Trans. Wireless
  Commun.}, vol.~6, no.~8, pp. 3114--3123, Aug. 2007.

\bibitem{Molisch2008}
R.~Madan, N.~Mehta, A.~Molisch, and J.~Zhang, ``Energy-efficient cooperative
  relaying over fading channels with simple relay selection,'' \emph{{IEEE}
  Trans. Wireless Commun.}, vol.~7, no.~8, pp. 3013--3025, Aug. 2008.

\bibitem{Karagiannidis2008}
D.~Michalopoulos and G.~Karagiannidis, ``Performance analysis of single relay
  selection in {R}ayleigh fading,'' \emph{{IEEE} Trans. Wireless Commun.},
  vol.~7, no.~10, pp. 3718--3724, Oct. 2008.

\bibitem{Yindi09a}
Y.~Jing and H.~Jafarkhani, ``Single and multiple relay selection schemes and
  their achievable diversity orders,'' \emph{{IEEE} Trans. Wireless Commun.},
  vol.~8, no.~3, pp. 1414--1423, Mar. 2009.

\bibitem{Yu07}
T.~C.-Y. Ng and W.~Yu, ``Joint optimization of relay strategies and resource
  allocations in cooperative cellular networks,'' \emph{{IEEE} J. Select. Areas
  Commun.}, vol.~25, no.~2, pp. 328--339, Feb. 2007.

\bibitem{Niu2010}
L.~Sun, T.~Zhang, L.~Lu, and H.~Niu, ``On the combination of cooperative
  diversity and multiuser diversity in multi-source multi-relay wireless
  networks,'' \emph{{IEEE} Signal Processing Lett.}, vol.~17, no.~6, pp.
  535--538, June 2010.

\bibitem{Ding2011}
H.~Ding, J.~Ge, D.~Benevides~da Costa, and Z.~Jiang, ``A new efficient
  low-complexity scheme for multi-source multi-relay cooperative networks,''
  \emph{{IEEE} Trans. Veh. Technol.}, vol.~60, no.~2, pp. 716--722, Feb. 2011.

\bibitem{Wittneben2010}
C.~Esli and A.~Wittneben, ``A hierarchical {AF} protocol for distributed
  orthogonalization in multiuser relay networks,'' \emph{{IEEE} Trans. Veh.
  Technol.}, vol.~59, no.~8, pp. 3902--3916, Oct. 2010.

\bibitem{Nosratinia07}
A.~Nosratinia and T.~E. Hunter, ``Grouping and partner selection in cooperative
  wireless networks,'' \emph{{IEEE} J. Select. Areas Commun.}, vol.~25, no.~2,
  pp. 369--378, Feb. 2007.

\bibitem{Xuemin2008}
J.~Cai, X.~Shen, J.~Mark, and A.~Alfa, ``Semi-distributed user relaying
  algorithm for amplify-and-forward wireless relay networks,'' \emph{{IEEE}
  Trans. Wireless Commun.}, vol.~7, no.~4, pp. 1348--1357, Apr. 2008.

\bibitem{Sharma2010}
S.~Sharma, Y.~Shi, Y.~T. Hou, and S.~Kompella, ``An optimal algorithm for relay
  node assignment in cooperative ad hoc networks,'' \emph{IEEE/ACM Trans.
  Networking}, vol. 19. , no. 3, pp. 879--892, June 2011.

\bibitem{Atapattu_ChinaCom}
S.~Atapattu, Y.~Jing, H.~Jiang, and C.~Tellambura, ``Opportunistic relaying in
  two-way networks,'' in \emph{Proc. 5th Int. {ICST} Conf. Commun. and
  Networking}, Aug. 2010.

\bibitem{Wang2003}
Z.~Wang and G.~B. Giannakis, ``A simple and general parameterization
  quantifying performance in fading channels,'' \emph{{IEEE} Trans. Commun.},
  vol.~51, no.~8, pp. 1389--1398, Aug. 2003.

\bibitem{Koyuncu2008}
E.~Koyuncu, Y.~Jing, and H.~Jafarkhani, ``Distributed beamforming in wireless
  relay networks with quantized feedback,'' \emph{{IEEE} J. Select. Areas
  Commun.}, vol.~26, no.~8, pp. 1429--1439, Oct. 2008.

\bibitem{gradshteyn2007}
I.~S. Gradshteyn and I.~M. Ryzhik, \emph{Table of Integrals, Series and
  Products}, 7th~ed.\hskip 1em plus 0.5em minus 0.4em\relax Academic Press Inc,
  2007.

\bibitem{Papoulis}
A.~Papoulis and Unnikrishna, \emph{Probability, Random Variables and Stochastic
  Processes with Errata Sheet}.\hskip 1em plus 0.5em minus 0.4em\relax {McGraw
  Hill Higher Education}, Jan. 2002.

\bibitem{Adve2006}
Y.~Zhao, R.~Adve, and T.~J. Lim, ``Symbol error rate of selection
  amplify-and-forward relay systems,'' \emph{{IEEE} Commun. Lett.}, vol.~10,
  no.~11, pp. 757--759, Nov. 2006.

\end{thebibliography}

\clearpage
\begin{figure}[ht]
\centering
\includegraphics[width=\figwidth]{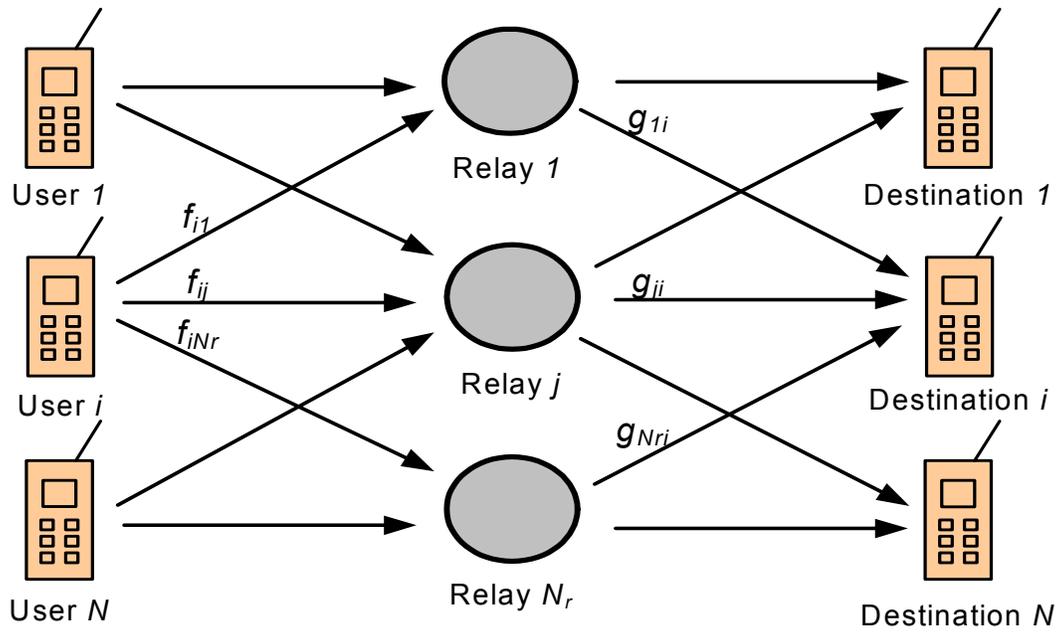}
\caption{A multiple-user relay network model.}
\label{fig_system}
\end{figure}

\begin{figure}[t]
\centering
\includegraphics[width=\figwidth]{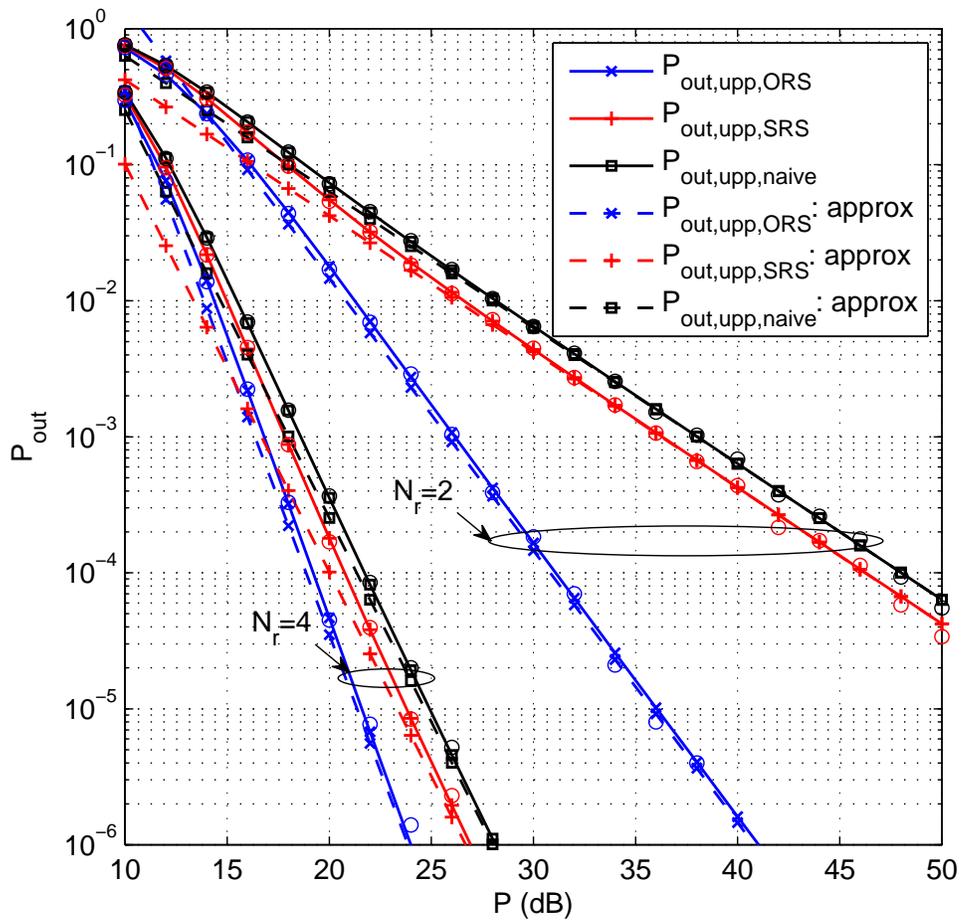}
\caption{Outage probability corresponding to $\gamma_{\min}$ for networks with two users and $N_r=2,4$ for ORS, SRS and naive RS schemes.}
\label{fig_verify}
\end{figure}

\begin{figure}[t]
\centering
\includegraphics[width=\figwidth]{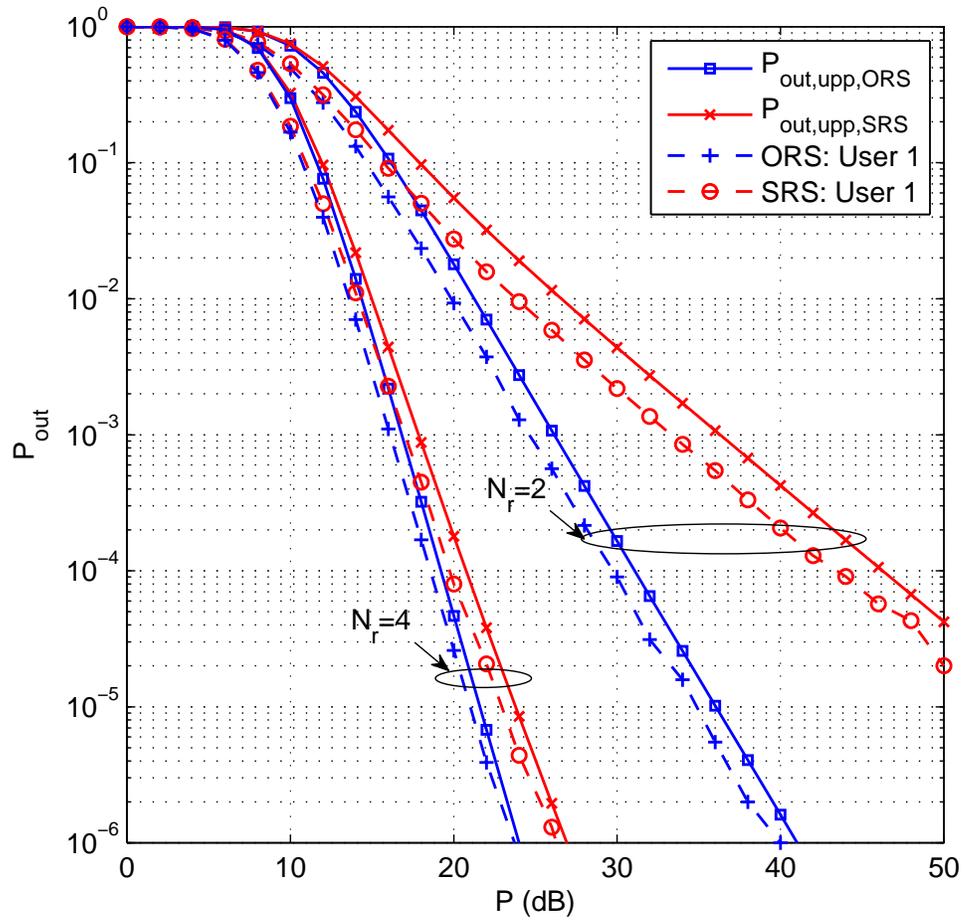}
\caption{Outage probabilities corresponding to $\gamma_{\min}$ and of users in networks with two users and $N_r=2,4$ for ORS and SRS schemes.}
\label{fig_Alg1Alg2}
\end{figure}

\begin{figure}[t]
\centering
\includegraphics[width=\figwidth]{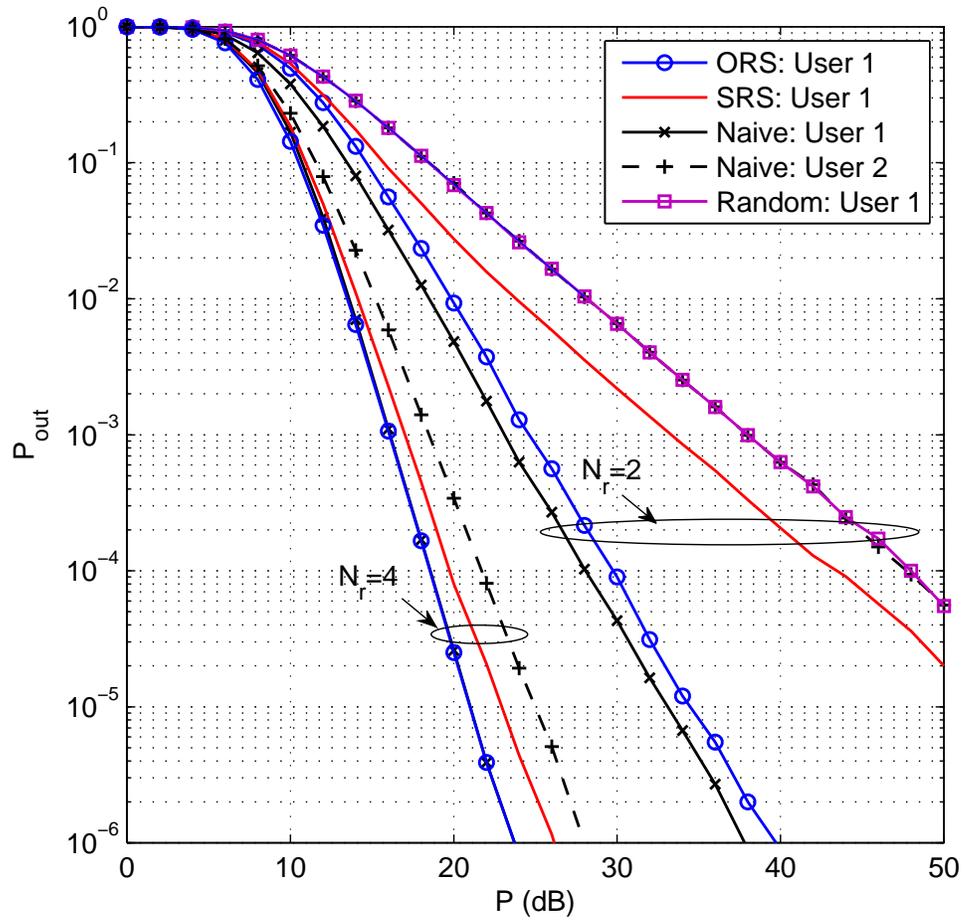}
\caption{Simulated outage probability for a network with two users and $N_r=2,4$ for ORS, SRS, naive and random RS scheme.}
\label{fig_allScheme1}
\end{figure}

\begin{figure}[t]
\centering
\includegraphics[width=\figwidth]{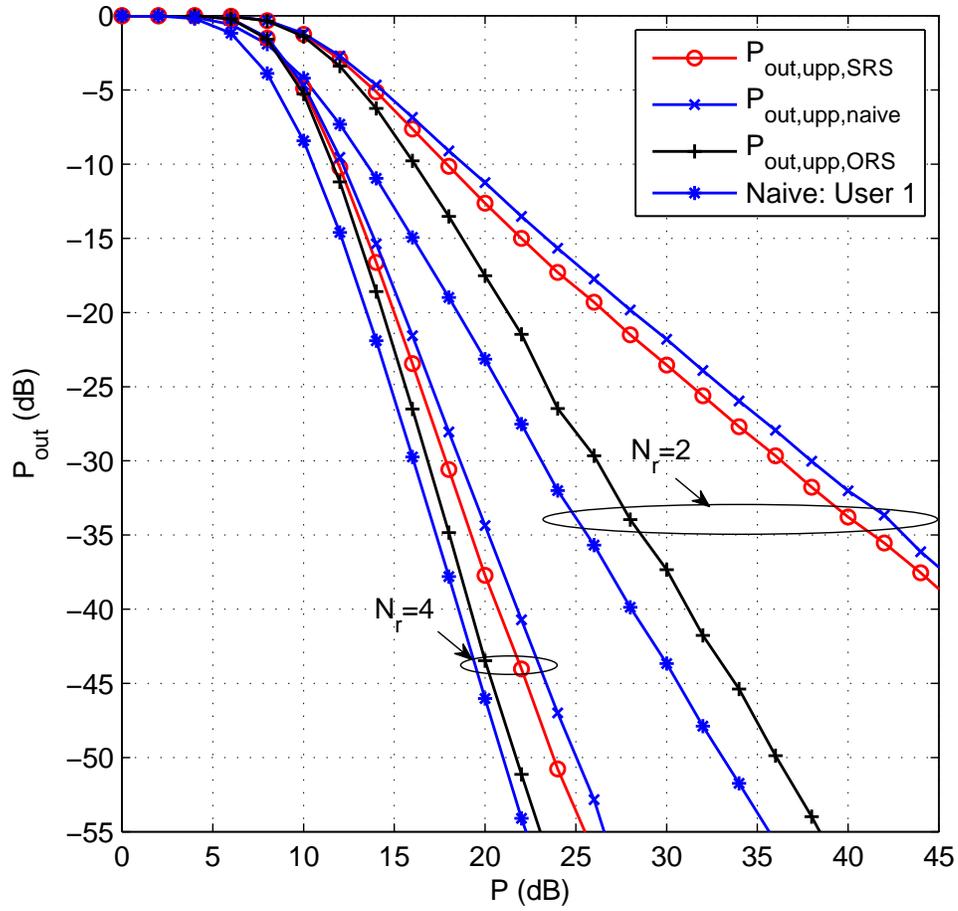}
\caption{Array gain difference between ORS and single-user best-relay case, and between SRS and the naive RS schemes.}
\label{fig_ArrNr}
\end{figure}

\begin{figure}[t]
\centering
\includegraphics[width=\figwidth]{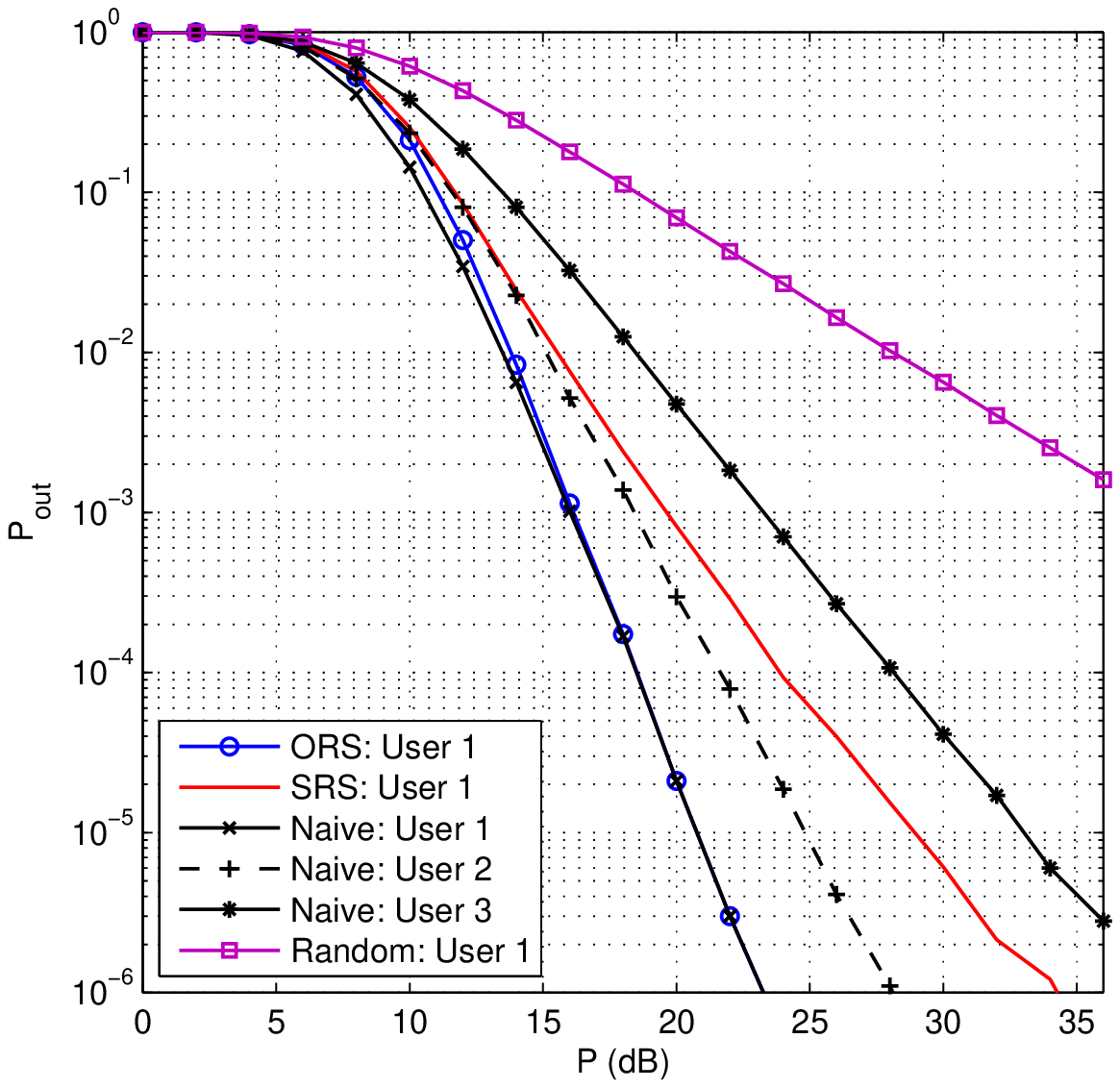}
\caption{Outage probability for a network with three users and $N_r=4$ for ORS, SRS, naive, and random RS schemes.}
\label{fig_multiuserNr4}
\end{figure}

\begin{figure}[t]
\centering
\includegraphics[width=\figwidth]{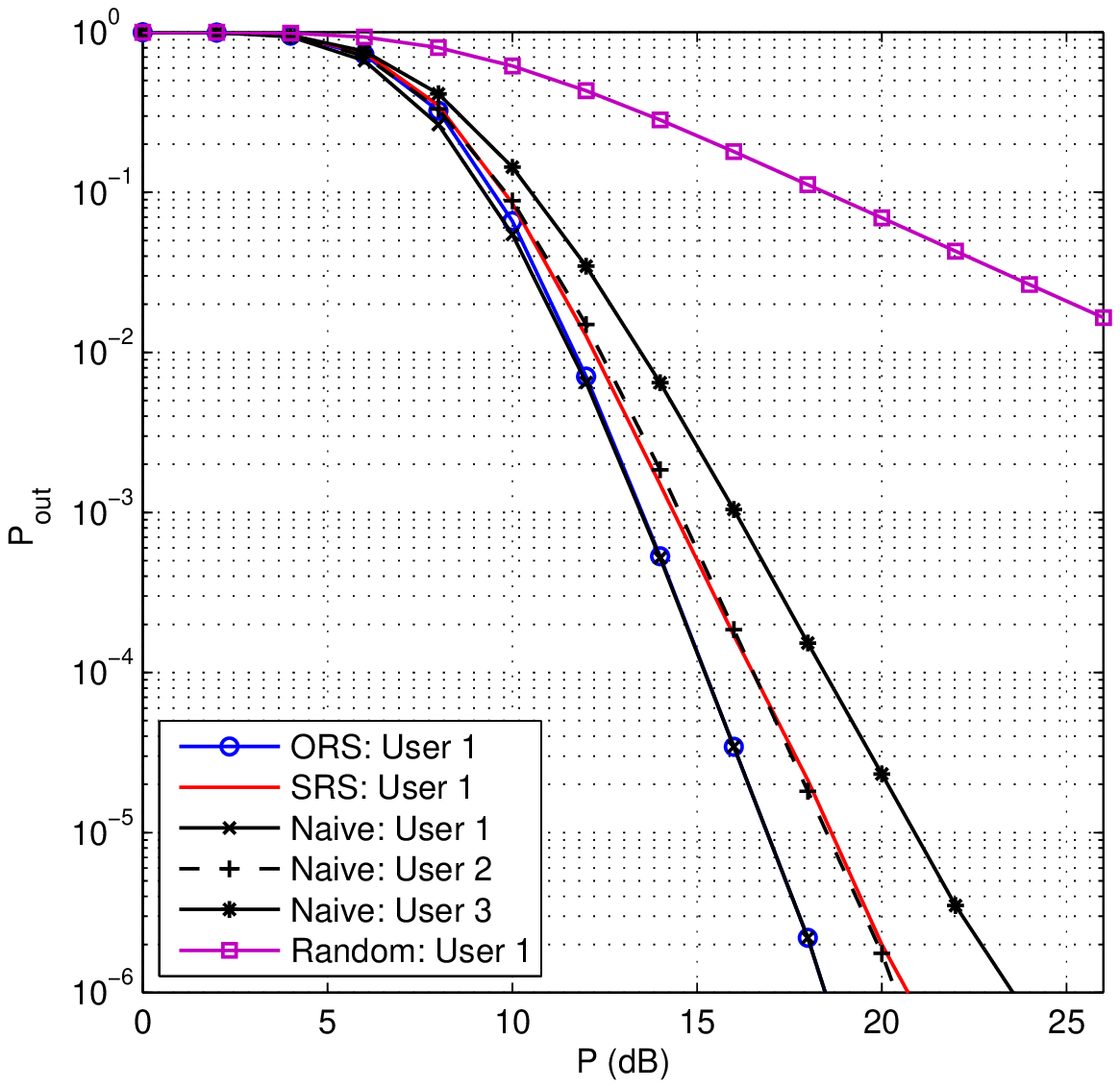}
\caption{Outage probability for a network with three users and $N_r=6$ for ORS, SRS, naive, and random RS schemes.}
\label{fig_multiuserNr6}
\end{figure}



\end{document}